\pdfoutput=1
\RequirePackage{ifpdf}
\ifpdf 
\documentclass[pdftex]{sigma} 
\else
\documentclass{sigma}
\fi

\def\a{\alpha}
\def\b{\beta}
\def\d{\delta}

\def\g{\gamma}
\def\l{\lambda}

\def\e{\epsilon}
\def\s{\sigma}

\numberwithin{equation}{section}

\newtheorem{Theorem}{Theorem}[section]
\newtheorem*{Theorem*}{Theorem}

\newtheorem{Proposition}[Theorem]{Proposition}
 { \theoremstyle{definition}

\newtheorem{Remark}[Theorem]{Remark} }

\begin{document}
\allowdisplaybreaks

\renewcommand{\PaperNumber}{074}

\FirstPageHeading

\ShortArticleName{The Generalized Lipkin--Meshkov--Glick Model and the Modified Algebraic Bethe Ansatz}

\ArticleName{The Generalized Lipkin--Meshkov--Glick Model\\ and the Modified Algebraic Bethe Ansatz}

\Author{Taras SKRYPNYK}

\AuthorNameForHeading{T.~Skrypnyk}

\Address{Bogolyubov Institute for Theoretical Physics, 14-b Metrolohichna Str., Kyiv, 03680, Ukraine}
\Email{\href{mailto:taras.skrypnyk@unimib.it}{taras.skrypnyk@unimib.it}}

\ArticleDates{Received June 19, 2022, in final form September 16, 2022; Published online October 10, 2022}

\Abstract{We show that the Lipkin--Meshkov--Glick $2N$-fermion model is a particular case of one-spin Gaudin-type model in an external magnetic field corresponding to a limiting case of non-skew-symmetric elliptic $r$-matrix and to an external magnetic field directed along one axis. We propose an exactly-solvable generalization of the Lipkin--Meshkov--Glick fermion model based on the Gaudin-type model corresponding to the same $r$-matrix but arbitrary external magnetic field. This model coincides with the quantization of the classical Zhukovsky--Volterra gyrostat. We diagonalize the corresponding quantum Hamiltonian by means of the modified algebraic Bethe ansatz. We explicitly solve the corresponding Bethe-type equations for the case of small fermion number $N=1,2$.}


\Classification{81R12; 82B23; 17B80}

\section{Introduction}
The Lipkin--Meshkov--Glick model was proposed in \cite{LMG} in order to describe shape phase transitions in nuclei. The corresponding $2N$-fermion Hamiltonian is written as follows:
\begin{gather}
\hat{H}_{\rm LMG} =  \frac{\epsilon}{2} \sum_{\sigma= \pm } \sum_{j=1}^N \sigma c^{\dag}_{j,\sigma}c_{j,\sigma} - \frac{W}{2} \sum_{i,j=1}^N\big( c^{\dag}_{i,+} c^{\dag}_{j,-}c_{j,+}c_{i,-} + c^{\dag}_{i,-} c^{\dag}_{j,+} c_{j,-}c_{i,+} \big) \nonumber\\
\hphantom{\hat{H}_{\rm LMG} =}{}-
\frac{V}{2} \sum_{i,j=1}^N \big(c^{\dag}_{i,+}c^{\dag}_{j,+} c_{j,-}c_{i,-} + c^{\dag}_{i,-} c^{\dag}_{j,-} c_{j,+}c_{i,+}\big), \label{hamLMG0}
\end{gather}
where $c_{j,\s'}$, $c^{\dag}_{i,\s}$, $i,j = 1,2,\dots ,N$, $\s,\s'\in \{+,-\}$ are fermion creation-anihilation operators,

The exact solvability of the model was shown in \cite{Pan} using a kind of Bethe ansatz technique. Later a connection of the Bethe ansatz for the Lipkin--Meshkov--Glick model with the Bethe ansatz for the trigonometric Gaudin model \cite{Gaud} was established in \cite{Duk}. The established connection is very indirect. In more details, the Lipkin--Meshkov--Glick Hamiltonian in the bosonic representation was connected \cite{Duk} with the combination of bosonized two-spin trigonometric Gaudin Hamiltonians in an external magnetic field directed along the third axis \cite{Lerma, Ortiz}. In such a way the Bethe ansatz existing for two-spin trigonometric Gaudin model was used to construct the Bethe ansatz for the Lipkin--Meshkov--Glick model.

The described above construction of \cite{Duk} seems to be somewhat artificial. Indeed, it is well-known \cite{Pan} the Hamiltonian~\eqref{hamLMG0} is rewritten in terms of~$\mathfrak{gl}(2)$ pseudo-spin operators $\hat{S}_{ij}$, \mbox{$i,j = 1,2$} as follows:
\begin{equation}
\hat{H} = \frac{\epsilon}{2} \big(\hat{S}_{11}- \hat{S}_{22}\big)-\frac{W}{2}\big(\hat{S}_{12}\hat{S}_{21}+ \hat{S}_{21}\hat{S}_{12}\big)-
\frac{V}{2}\big( \hat{S}_{12}^2 + \hat{S}_{21}^2\big), \label{hamGinS0'}
\end{equation}
where the commutation relations of the pseudo-spin operators $\hat{S}_{ij}$, $i,j = 1,2$, are the following ones:
\begin{equation}\label{pso}
\big[\hat{S}_{ij}, \hat{S}_{kl}\big]= \d_{kj} \hat{S}_{il}- \d_{il} \hat{S}_{kj}.
\end{equation}
This Hamiltonian should be connected with the class of integrable one-spin (not two-spin!) models.

In terms of the standard basis of $\mathfrak{so}(3)\simeq \mathfrak{sl}(2)$ the Hamiltonian~\eqref{hamGinS0'} is rewritten as follows:
\begin{gather}
\hat{H} = {\rm i} \e \hat{S}_{3} + (W+V) \hat{S}_{1}^2+ (W-V) \hat{S}_{2}^2, \qquad {\rm i}=\sqrt{-1}, \label{hamGinS0}
\end{gather}
where $\hat{S}_{\a}$, $\a = 1,2,3$, satisfy the standard commutation relations of $\mathfrak{so}(3)$ pseudo-spin operators:
\begin{gather}
\big[\hat{S}_{\a}, \hat{S}_{\b}\big]= \e_{\a\b\g} \hat{S}_{\g}. \label{sscr}
\end{gather}

We remind that standard Gaudin integrable spin chain Hamiltonians \cite{Gaud}, i.e., the ones based on the classical skew-symmetric $r$-matrices \cite{Skl}, contain spin-spin interaction terms for different spins of the chain and do not contain one spin self-interaction terms. That is why it is impossible to obtain the Hamiltonian~\eqref{hamGinS0} using the standard Gaudin Hamiltonians \cite{Gaud} with one spin. On the other hand, in the series of our previous papers \cite{SkrPLA2005, SkrJGP2006, SkrJPA2007} we have proposed a generalization of the Gaudin Hamiltonians
 based on the non-skew-symmetric classical $r$-matrices $r(u_1,u_2)=
\sum_{\a,\b=1}^{3}r_{\a\b}(u_1, u_2)X_{\a} \otimes
X_{\b}$, where $X_{\a}$ is a standard basis in $\mathfrak{so}(3)$, satisfying generalized classical Yang--Baxter equation
 \cite{AT1,BabVia, Maillet} instead of the usual one.
 In one-spin-case our generalized Gaudin Hamiltonian has the following form \cite{SkrPLA2005, SkrJGP2006, SkrJPA2007}:
 \begin{equation*}
\hat{H}= \frac{1}{2}\sum_{\a,\b=1}^3 r^0_{\a\b}(\nu, \nu)\big(\hat{S}_{\a} \hat{S}_{\b}+ \hat{S}_{\b} \hat{S}_{\a}\big)+ \sum_{\a=1}^3 c_{\a}(\nu) \hat{S}_{\a}, 
\end{equation*}
 where $r^0_{\a\b}(\nu, \nu)$ are the components of the regular part of the classical $r$-matrix $r(u_1,u_2)$, $\nu$~is a~fixed non-singular \cite{SkrJMP2016} value of spectral parameter and $c_{\a}(\nu)$ are the components of the so-called shift element playing the role of an external field \cite{SkrJPA2007, SkrJGP2014}.

In the present short paper we show, that the Hamiltonian~\eqref{hamGinS0} is a one-spin generalized Gaudin Hamiltonian connected with the following $r$-matrix:
\begin{gather}
r(u,v)= \frac{\sqrt{u+j_{2}}\sqrt{v+j_{1}}}{u-v} X_{1}\otimes X_{1} + \frac{\sqrt{u+j_{1}}\sqrt{v+j_{2}}}{u-v} X_{2}\otimes X_{2}\nonumber \\
\hphantom{r(u,v)=}{}+ \frac{\sqrt{u+j_{1}}\sqrt{u+j_{2}}}{u-v} X_{3}\otimes X_{3}, \label{twellrmlim0}
\end{gather}
which is a $j_3\rightarrow \infty$ limit of the non-skew-symmetric elliptic $r$-matrix \cite{SkrPLA2005, SkrJMP2006, SkrJGP2006} defined on the $4:1$ unramified covering of the Weierstrass cubic $y^2=(u+j_1)(u+j_2)(u+j_3)$.

 The $r$-matrix~\eqref{twellrmlim0} possesses the following shift element defining integrable external magnetic field:
 \begin{equation*}
 c(u)= \sum_{\a=1}^2 \frac{ {\rm i} k_{\a}}{2 \sqrt{u+j_{\a}}}X_{\a}+ \frac{{\rm i} k_3}{2} X_3, 
 \end{equation*}
 where the constants $k_{\a}$, $\a = 1,2,3$, are arbitrary.
 It leads to the following quantum Hamiltonian:
 \begin{gather}
\hat{H} =
 {\rm i} k_{1} \sqrt{j_{2}}\hat{S}_{1} + {\rm i} k_{2} \sqrt{j_{1}}\hat{S}_{2}+
 {\rm i} k_3 \sqrt{j_{1}} \sqrt{j_{2}} \hat{S}_{3} - j_2 \hat{S}_{1}^2 -j_1 \hat{S}_{2}^2. \label{hamgGc00}
\end{gather}
We call it generalized Lipkin--Meshkov--Glick Hamiltonian.
 In the case $k_1=k_2=0$ it coincides~-- modulo the renaming of the constants~-- with the standard Lipkin--Meshkov--Glick Hamiltonian~\eqref{hamGinS0}.

Note also that the Hamiltonian~\eqref{hamgGc00} coincide -- up to the quadratic Casimir operator $\hat{C}_2= \hat{S}_{1}^2+\hat{S}_{2}^2+ \hat{S}_{3}^2$ and renaming of the constants -- with quantization of Zhukovsky--Volterra Hamiltonian \cite{Volt} which is written explicitly as follows:
\[
\hat{H} = b_{1} \hat{S}_{1} + b_{2} \hat{S}_{2}+ b_3 \hat{S}_{3} + a_1 \hat{S}_{1}^2 +a_2 \hat{S}_{2}^2+ a_3 \hat{S}_{3}^2.
\]
Here the constant $a_i$, $b_i$, $i= 1,2,3$, are arbitrary.

 We use the discovered connection of the generalized Lipkin--Meshkov--Glick Hamiltonian with the $r$-matrix~\eqref{twellrmlim0} in order to find its spectrum by means of the modified algebraic Bethe ansatz (mABA). Indeed in our paper \cite{SkrNPB2022} we have constructed the spectrum and Bethe vectors for the generalized Gaudin models with $N$ spins associated with the $r$-matrix~\eqref{twellrmlim0} using modified algebraic Bethe ansatz. In \cite{SkrNPB2022} we have applied the obtained results to find the spectrum of the Richardson-type model which is $N$-spin case of the Gaudin-type model with the same simplest fundamental representation of each one-spin $\mathfrak{sl}(2)\simeq \mathfrak{so}(3)$ algebra in the chain of $N$ spins. In the present letter we apply the mABA to one-spin case of the Gaudin-type model and arbitrary representation of the corresponding one-spin algebra and obtain the following spectrum of the Hamiltonian~\eqref{hamgGc00}:
 \begin{equation*}
h_{\pm}(v_1, v_2, \dots , v_N)=2 N \left((j_1+j_2) N \pm\frac{{\rm i}k_1 j_{2}+ k_2 j_{1}}{\sqrt{j_1-j_2}}+2 j_1 j_2 \sum_{k=1}^N \frac{1}{v_k} \right), 
\end{equation*}
where the rapidities $v_k$ satisfy the following set of modified Bethe equations:
 \begin{gather}
 -\frac{N (v_k+j_1)(v_k+j_2)}{v_k} \mp \frac{{\rm i}k_1 (v_k+j_{2})+ k_2 (v_k+j_{1})}{\sqrt{j_1-j_2}}\nonumber \\
 \qquad \quad{}+(2 v_k+j_1+j_2) + 2\sum_{n=1, \, n\neq k}^N \frac{(v_k+j_1)(v_k+j_2)}{v_k-v_n} \nonumber \\
 \qquad{}= \frac{1}{4}\big(k_+ \mp( N +1)\big)\big(k_- \mp( N +1)\big) \frac{v_k^{N}}{\prod\limits_{n=1,\, n\neq k}^{N}(v_k-v_n)}, \qquad k = 1,\dots ,N, \label{betheeq0}
\end{gather}
 corresponding to two sets of the modified Bethe vectors. Here $k_{\pm}=\frac{{\rm i} k_1+ k_2}{\sqrt{j_1-j_2}} \pm k_3$ and the number~$N$ is related with the representation space of $\mathfrak{sl}(2)\simeq \mathfrak{so}(3)$ and coincide with the fermion number.

 It is necessary to notice that for generic values of $k_i$ both sets of modified Bethe equations~\eqref{betheeq0} seem to produce the same modified Bethe vectors and corresponding set of the eigenvalues $\{h_{+}(v_1, v_2, \dots , v_N)\}$ and $\{h_{-}(v_1, v_2, \dots , v_N)\}$ coincide. For illustration purposes we solve the modified Bethe equations~\eqref{betheeq0} and explicitly find the spectrum of the Hamiltonian~\eqref{hamgGc00} in the cases $N=1,2$.

 In the end of the introduction we remark, that the modified Bethe ansatz in the context of the rational Gaudin model in arbitrary magnetic field has appeared in the paper \cite{FafibaultT}. In the context of the ``shifted'' trigonometric $r$-matrix and the corresponding Richardson-type models interacting with the environment the mABA has appeared in the papers \cite{Claeys, LukLinks}. We also remark that integrable models associated with certain limits of the non-skew-symmetric elliptic $r$-matrix have been considered also in the recent papers \cite{DF, SIL}.

 The structure of the present paper is the following: in Section~\ref{section2} we describe the $j_3\rightarrow \infty$ limit of the non-skew-symmetric elliptic $r$-matrix, in Section~\ref{section3} we describe its Lax algebra and one-spin Gaudin-type model. In Section~\ref{section4} we consider its fermionization and obtain the generalized LMG model. Section~\ref{section5} is devoted to the modified Bethe ansatz. In Section~\ref{section6} we conclude and discuss the open problems.

\section[Elliptic r-matrix in the j\_3 to infty limit]{Elliptic $\boldsymbol{r}$-matrix in the $\boldsymbol{j_3\rightarrow \infty}$ limit} \label{section2}
\subsection[Non-skew-symmetric elliptic r-matrix in so(3) basis]{Non-skew-symmetric elliptic $\boldsymbol{r}$-matrix in $\boldsymbol{\mathfrak{so}(3)}$ basis}
Let us consider the following tensor \cite{SkrPLA2005,SkrJGP2006,SkrNPB2019}:
\begin{equation}
r(u,v)=\sum_{\a=1}^3 \frac{u_{1}u_{2}u_{3}}{u-v} \frac{v_{\a}}{u_{\a}} X_{\a}\otimes X_{\a}, \label{twellrm}
\end{equation}
where $\{ X_{\a} \mid \a= 1,2,3\}$ is a basis over $\mathbb{C}$ in $\mathfrak{sl}(2)\simeq \mathfrak{so}(3)$ with the commutation relations repeating the commutation relations~\eqref{sscr}
and the functions $u_{\a}$, $v_{\a}$, $\a= 1,2,3$, are defined as follows:
\[
 u^2_{\a}=u+j_{\a}, \qquad v^2_{\a}=v+j_{\a}, \qquad \a= 1,2,3.
 \]
It is possible to show that~\eqref{twellrm} satisfies generalized classical Yang--Baxter equation \cite{AT1,BabVia, Maillet}:
\begin{gather*}
\big[r^{12}(u_1,u_2), r^{13}(u_1,u_3)\big]= \big[r^{23}(u_2,u_3),
r^{12}(u_1,u_2)\big]-\big[r^{32}(u_3,u_2), r^{13}(u_1,u_3)\big]. 
\end{gather*}
Here
\begin{gather*}
r^{12}(u_1,u_2)\equiv
\sum_{\a,\b=1}^{3}r_{\a\b}(u_1, u_2)X_{\a} \otimes
X_{\b}\otimes 1,\\
r^{13}(u_1,u_3)\equiv
\sum_{\a,\b=1}^{3}r_{\a\b}(u_1, u_3)X_{\a} \otimes 1
\otimes X_{\b}, \qquad \text{etc.}
\end{gather*}

The $r$-matrix~\eqref{twellrm} satisfies -- up to the overall multiplier $v_1 v_2 v_3$ -- the following condition:
\begin{gather}
r(u,v)=\frac{\Omega}{u-v}+ r_0(u,v), \label{reg}
\end{gather}
where $\Omega=\sum_{\a=1}^3 X_{\a}\otimes X_{\a}$ and $r_0(u,v)$ is a regular on the diagonal $u=v$ function.

An important characteristic of the $r$-matrix is the existence of the ``shift element'', i.e., \mbox{$c$-valued function}
$c(u)=\sum_{\a=1}^3 c_{\a}(u)X_{\a}$ satisfying the ``shift equation'' \cite{SkrJPA2007}:
\begin{equation*}
\big[r^{12}(u_1,u_2),c^{(1)}(u_1)\big]- \big[r^{21}(u_2,u_1),c^{(2)}(u_2)\big]=0, 
\end{equation*}
where $c^{(1)}(u_1)=c(u_1)\otimes 1$, $c^{(2)}(u_2)=1\otimes c(u_2)$.

 The following Proposition holds true \cite{SkrJPA2007,SkrJGP2014,SkrNPB2019}:
 \begin{Proposition}
 The following $\mathfrak{sl}(2)\simeq \mathfrak{so}(3)$-valued function is a shift element for the $r$-mat\-rix~\eqref{twellrm}:
 \begin{equation}
 c(u)=\sum_{\a=1}^3 \frac{ {\rm i} k_{\a}}{2 u_{\a}}X_{\a}, \label{shelel}
 \end{equation}
 where the constants $k_{\a}$, $\a= 1,2,3$, are arbitrary.
 \end{Proposition}

\subsection[The j\_3 to infty limit]{The $\boldsymbol{j_3\rightarrow \infty}$ limit}
\subsubsection[The so(3) basis]{The $\boldsymbol{\mathfrak{so}(3)}$ basis}

Let us consider the limit $j_3\rightarrow \infty$ in the $r$-matrix~\eqref{twellrm} rescaled by $\frac{1}{\sqrt{j_3}}$.
We will have
\begin{equation}
r(u,v)=\frac{u_2 v_1}{u-v} X_{1}\otimes X_{1} + \frac{u_1 v_2}{u-v} X_{2}\otimes X_{2} + \frac{u_1u_2}{u-v} X_{3}\otimes X_{3}.\label{twellrmlim}
\end{equation}
It is easy to show that this $r$-matrix satisfies -- up to the factor $v_1 v_2$ -- the condition~\eqref{reg}.

Let us now find what happens with the shift element~\eqref{shelel} in this limit.
 \begin{Proposition}[\cite{SkrNPB2019}]
 The following $\mathfrak{so}(3)$-valued function is a shift element for the $r$-mat\-rix~\eqref{twellrmlim}:
 \begin{equation}
 c(u)= \sum_{\a=1}^2 \frac{ {\rm i} k_{\a}}{2 u_{\a}}X_{\a}+ \frac{{\rm i} k_3}{2} X_3, \label{shelellim}
 \end{equation}
 where the constants $k_{\a}$, $\a= 1,2,3$, are arbitrary.
 \end{Proposition}

\begin{Remark}Rigorously speaking, in the limit $j_3\rightarrow \infty$ the curve $y^2=(u+j_1)(u+j_2)(u+j_3)$ 
 stops to be elliptic.
 Nevertheless the $r$-matrix~\eqref{twellrmlim} and the corresponding integrable systems are still completely anisotropic, since all anisotropy parameters $j_{\a}$, $\a= 1,2,3$, are not equal in the considered case.
\end{Remark}

\subsubsection[The gl(2) basis]{The $\boldsymbol{\mathfrak{gl}(2)}$ basis}
Let us now use the connection of the $\mathfrak{so}(3)$ basis in~$\mathfrak{sl}(2)$ with the standard~$\mathfrak{gl}(2)$ basis $\{ X_{ij} \mid i,j = 1,2\}$:
\begin{gather*}
 X_1 = -\frac{{\rm i}}{2}(X_{21}+X_{12}),\qquad \! X_2 = -\frac{1}{2}(X_{12}-X_{21}), \qquad \! X_3 = -\frac{{\rm i}}{2}(X_{11}-X_{22}), \qquad \! {\rm i}=\sqrt{-1}, 
\end{gather*}
where the commutation relations among $X_{ij}$, $X_{kl}$ repeat the commutation relations~\eqref{pso}.

In the~$\mathfrak{gl}(2)$ basis the $r$-matrix~\eqref{twellrmlim} acquires -- up to the coefficient $\frac{1}{4}$ -- the following explicit form:
\begin{gather}
r(u,v) =   \frac{u_1 u_2 (X_{11}-X_{22})\otimes (X_{11}-X_{22})}{(v-u)}+\frac{(v_1u_2+u_1v_2)(X_{12}\otimes X_{21}+X_{21}\otimes X_{12})}{(v-u)} \nonumber \\
\hphantom{r(u,v) =}{}+\frac{(v_1u_2-u_1v_2)(X_{12}\otimes X_{12}+X_{21}\otimes X_{21})}{(v-u)}. \label{twellrmlimgl2-0}
\end{gather}
The $r$-matrix~\eqref{twellrmlimgl2-0} satisfies -- up to the multiplier $(- 2 v_1 v_2)$ -- the condition~\eqref{reg} with
\[
\Omega= \frac{1}{2} (X_{11}-X_{22})\otimes (X_{11}-X_{22}) + X_{12}\otimes X_{21}+X_{21}\otimes X_{12}.
\]
Instead of the $\mathfrak{sl}(2)\otimes \mathfrak{sl}(2)$-valued $r$-matrix $r(u,v)$ given by the formula~\eqref{twellrmlimgl2-0} it will be convenient to consider the following equivalent $\mathfrak{gl}(2)\otimes \mathfrak{gl}(2)$-valued $r$-matrix:
 \begin{gather}
r(u,v) =  \frac{2 u_1 u_2 (X_{11}\otimes X_{11} +X_{22}\otimes X_{22})}{(v-u)}+\frac{(v_1u_2+u_1v_2)(X_{12}\otimes X_{21}+X_{21}\otimes X_{12})}{(v-u)} \nonumber\\
\hphantom{r(u,v) =}{}+\frac{(v_1u_2-u_1v_2)(X_{12}\otimes X_{12}+X_{21}\otimes X_{21})}{(v-u)}. \label{twellrmlimgl2}
\end{gather}
The $r$-matrix~\eqref{twellrmlimgl2} satisfies -- up to the multiplier $(-2 v_1v_2)$ -- the condition~\eqref{reg} with
\[
\Omega= X_{11}\otimes X_{11} +X_{22}\otimes X_{22}+ X_{12}\otimes X_{21}+X_{21}\otimes X_{12}.
\]
The shift element~\eqref{shelellim} is rewritten in terms of the~$\mathfrak{gl}(2)$ basis -- up to the coefficient $\frac{1}{4}$ -- as follows:
\begin{equation}
 c(u)= k_3(X_{11}-X_{22})+ \bigg(\frac{k_{1}}{ u_1}- \frac{{\rm i} k_{2}}{u_2 }\bigg)X_{12}+ \bigg(\frac{k_{1}}{ u_1}+ \frac{{\rm i} k_{2}}{u_2}\bigg)X_{21}. \label{shelellimgl2}
 \end{equation}
It satisfies the shift equation for both $r$-matrices~\eqref{twellrmlimgl2-0} and~\eqref{twellrmlimgl2}.
\section{Lax algebra, generating function and one-spin model } \label{section3}
\subsection[Lax algebra in the j\_3 to infty limit: the general case]{Lax algebra in the $\boldsymbol{j_3 \rightarrow \infty}$ limit: the general case}
\subsubsection[The so(3) basis and generating function of the quantum integrals]{The $\boldsymbol{\mathfrak{so}(3)}$ basis and generating function of the quantum integrals}
Let us at first consider the Lax matrix corresponding to the $r$-matrix~\eqref{twellrmlim} in the natural $\mathfrak{so}(3)$-basis:
\begin{equation*}
\hat{L}(u)=\sum_{\a=1}^3 \hat{L}_{\a}(u) X_{\a}.
\end{equation*}
The corresponding Lax algebra
\begin{equation*}
\big[\hat{L}^{(1)}(u_1),\hat{L}^{(2)}(u_2)\big]=\big[r^{12}(u_1,u_2),\hat{L}^{(1)}(u_1)\big]-
\big[r^{21}(u_2,u_1),\hat{L}^{(2)}(u_2)\big], 
\end{equation*}
where $\hat{L}^{(1)}(u_1)= \hat{L}(u_1)\otimes 1$, $\hat{L}^{(2)}(u_2)=1 \otimes \hat{L}(u_2)$,
 has very simple component form:
 \begin{gather*}
\big[\hat{L}_{1}(u),\hat{L}_{2}(v)\big]= \frac{ u_1 v_2 }{(u-v)} \big(\hat{L}_3(u)- \hat{L}_3(v)\big), \\
\big[\hat{L}_{1}(u),\hat{L}_{3}(v)\big]= -\frac{ u_2 u_1 }{(u-v)} \hat{L}_2(u)+ \frac{ u_1 v_2}{(u-v)} \hat{L}_2(v), \\
\big[\hat{L}_{2}(u),\hat{L}_{3}(v)\big]= \frac{ u_2 u_1 }{(u-v)} \hat{L}_1(u)- \frac{ u_2 v_1}{(u-v)} \hat{L}_1(v), \\
\big[\hat{L}_{1}(u),\hat{L}_{1}(v)\big]= \big[\hat{L}_{2}(u),\hat{L}_{2}(v)\big]= \big[\hat{L}_{3}(u),\hat{L}_{3}(v)\big]=0.
\end{gather*}

Let us consider the following quadratic functions in the elements of the Lax algebra:
\begin{equation}
\hat{\tau}^{(2)}(u)=-\big(\hat{L}^2_{1}(u)+ \hat{L}^2_{2}(u)+\hat{L}^2_{3}(u)\big). \label{tauso3}
\end{equation}
As it follows from the general results of \cite{SkrJMP2007} the generating function $\hat{\tau}^{(2)}(u)$, $\hat{\tau}^{(2)}(v)$ commute
\[
\big[\hat{\tau}^{(2)}(u), \hat{\tau}^{(2)}(v)\big]=0,
\]
since the $r$-matrix -- modulo the overall multiplier $v_1 v_2$ -- satisfy the condition~\eqref{reg}.

\subsubsection[The gl(2) basis and the generating function of quantum integrals]{The $\boldsymbol{\mathfrak{gl}(2)}$ basis and the generating function of quantum integrals}
Let us consider the Lax matrix that correspond to the $r$-matrix~\eqref{twellrmlimgl2} in the natural~$\mathfrak{gl}(2)$-basis:
\begin{gather}
\hat{L}(u)=\sum_{i,j=1}^2 \hat{L}_{ij}(u)X_{ij}. \label{Lax}
\end{gather}
The corresponding Lax algebra
\[
\big[\hat{L}^{(1)}(u_1),\hat{L}^{(2)}(u_2)\big]=\big[r^{12}(u_1,u_2),\hat{L}^{(1)}(u_1)\big]-
\big[r^{21}(u_2,u_1),\hat{L}^{(2)}(u_2)\big],
\]
where $\hat{L}^{(1)}(u_1)= \hat{L}(u_1)\otimes 1$, $\hat{L}^{(2)}(u_2)=1 \otimes \hat{L}(u_2)$,
 has the following component form in this basis:
\begin{subequations}\label{laxalj3inf}
\begin{gather}
\big[\hat{L}_{11}(u),\hat{L}_{12}(v)\big]= \frac{ (u_2 v_1+ u_1 v_2)}{(u-v)} \hat{L}_{12}(u) +\frac{(u_1 v_2-u_2 v_1)}{(u-v)} \hat{L}_{21}(u)-\frac{2 v_1 v_2}{(u-v)} \hat{L}_{12}(v),
\\
\big[\hat{L}_{11}(u),\hat{L}_{21}(v)\big]=\frac{ (u_2 v_1 - u_1 v_2)}{(u-v)} \hat{L}_{12}(u) -\frac{(u_2 v_1+u_1v_2)}{(u-v)} \hat{L}_{21}(u)+\frac{2 v_1 v_2}{(u-v)} \hat{L}_{21}(v),
\\
\big[\hat{L}_{12}(u),\hat{L}_{21}(v)\big]= \frac{ (u_2 v_1+u_1 v_2)}{(u-v)}\big(\hat{L}_{11}(u)- \hat{L}_{22}(u)- \hat{L}_{11}(v)+ \hat{L}_{22}(v)\big),
\\
\big[\hat{L}_{12}(u),\hat{L}_{12}(v)\big]= \frac{ (u_2 v_1-u_1 v_2)}{(u-v)}\big(\hat{L}_{11}(u)- \hat{L}_{22}(u)- \hat{L}_{11}(v)+ \hat{L}_{22}(v)\big),
\\
\big[\hat{L}_{21}(u),\hat{L}_{21}(v)\big]= - \frac{ (u_2 v_1-u_1 v_2)}{(u-v)}\big(\hat{L}_{11}(u)- \hat{L}_{22}(u)- \hat{L}_{11}(v)+ \hat{L}_{22}(v)\big),
\\
\big[\hat{L}_{22}(u),\hat{L}_{12}(v)\big]=-\big[\hat{L}_{11}(u),\hat{L}_{12}(v)\big] ,
\\
\big[\hat{L}_{22}(u),\hat{L}_{21}(v)\big]=-\big[\hat{L}_{11}(u),\hat{L}_{21}(v)\big],
\\
\big[\hat{L}_{11}(u),\hat{L}_{11}(v)\big]= \big[\hat{L}_{11}(u),\hat{L}_{22}(v)\big]= \big[\hat{L}_{22}(u),\hat{L}_{22}(v)\big]=0.
\end{gather}
\end{subequations}

From the commutation relations~\eqref{laxalj3inf} immediately follows that
\[
\hat{\tau}^{(1)}(u)=\hat{L}_{11}(u)+\hat{L}_{22}(u)
\]
 is a central element of the Lax algebra.

The generating function of the quadratic integrals is written in terms of~$\mathfrak{gl}(2)$ basis as follows:
\begin{gather}
\hat{\tau}^{(2)}(u)=\frac{1}{2}\bigl(\hat{L}^2_{11}(u)+ \hat{L}^2_{22}(u)+ \hat{L}_{12}(u)\hat{L}_{21}(u)+ \hat{L}_{21}(u)\hat{L}_{12}(u)\bigr). \label{taugl2}
\end{gather}

Since the $r$-matrix -- modulo the multiplier $-2 v_1v_2$ -- satisfies the condition~\eqref{reg} and $\hat{\tau}^{(1)}(u)$ is a central element of its Lax algebra,
the generating functions $\hat{\tau}^{(2)}(u)$, $\hat{\tau}^{(2)}(v)$ commute \cite{SkrNPB2021}:
\[
\big[\hat{\tau}^{(2)}(u), \hat{\tau}^{(2)}(v)\big]=0.
\]
 The generating function $\hat{\tau}^{(2)}(u)$ given by~\eqref{taugl2} coincide with the generating function $\hat{\tau}^{(2)}(u)$ given by~\eqref{tauso3} up to the square of the central element, namely, up to
 $\big(\frac{\hat{\tau}^{(1)}(u)}{2}\big)^2= \frac{1}{4} \big(\hat{L}_{11}(u)+ \hat{L}_{22}(u)\big)^2$.

\begin{Remark}
 In what follows  we will, slightly abusing the algebro-geometric language, use the explicit formulae for $u_{\a}$, $v_{\a}$ in terms of square roots:
 \[
 u_{\a}=\sqrt{u+j_{\a}},\qquad v_{\a}=\sqrt{v+j_{\a}}, \qquad \a= 1,2,3.
 \]
\end{Remark}

\subsection{Lax algebra and Hamiltonian: case of the one-spin model}

\subsubsection[The so(3) basis]{The $\boldsymbol{\mathfrak{so}(3)}$ basis}
Let us explicitly describe the Lax matrix of a quantum spin in the external magnetic field corresponding to the considered $r$-matrix in the forms~\eqref{twellrmlim} and~\eqref{twellrmlimgl2-0}.

We start with the form~\eqref{twellrmlim} corresponding to the $\mathfrak{so}(3)$ basis. As it follows from the general theory \cite{SkrPLA2005, SkrJGP2006, SkrJPA2007, SkrNPB2022}, the Lax matrix with the first order pole in the point $u=0$ is written as follows:
\begin{equation*}
\hat{L}(u)= -\bigg(\frac{\sqrt{j_{2}}\sqrt{u+j_{1}}}{u} \hat{S}_{1} X_{1} + \frac{\sqrt{j_{1}}\sqrt{u+j_{2}}}{u} \hat{S}_{2} X_{2} + \frac{\sqrt{j_{1}}\sqrt{j_{2}}}{u} \hat{S}_{3} X_{3}\bigg) +c(u),
\end{equation*}
where $\hat{S}_{\a}$ is the $\a$-th component of the spin operator with the commutation relations~\eqref{sscr},
 the shift element $c(u)$ is given by the formula~\eqref{shelellim}.

The corresponding integrable quantum Hamiltonian is obtained from the generating func\-tion~$\hat{\tau}^{(2)}(u)$:
\begin{equation*}
\hat{\tau}^{(2)}(u)=-\big(\hat{L}^2_1(u)+ \hat{L}^2_2(u)+ \hat{L}^2_3(u)\big).
\end{equation*}
In more details, we will consider the following Hamiltonian:
\[
\hat{H}= \mathrm{res}_{u=0} \hat{\tau}^{(2)}(u).
\]
The direct calculation shows that it has the form~\eqref{hamgGc00}.

\subsubsection[The gl(2) basis]{The $\boldsymbol{\mathfrak{gl}(2)}$ basis}
Let us now rewrite the Lax matrix and Gaudin-type Hamiltonians in an external magnetic field in~$\mathfrak{gl}(2)$ basis corresponding to the $r$-matrix~\eqref{twellrmlimgl2-0}.
It has the form~\eqref{Lax} with the following components:
\begin{subequations}\label{onespinlax}
\begin{gather}
\hat{L}_{11}(u) = \frac{2\sqrt{j_1}\sqrt{j_2}\hat{S}_{11}}{u}+k_3,
\\
 \hat{L}_{22}(u) = \frac{2\sqrt{j_1}\sqrt{j_2} \hat{S}_{22}}{u}-k_3,
\\
\hat{L}_{12}(u) = \frac{\sqrt{u+j_1}\sqrt{j_2}\big(\hat{S}_{21}+ \hat{S}_{12}\big)+\sqrt{j_1}\sqrt{u+j_2}\big(\hat{S}_{21}-\hat{S}_{12}\big)}{u}\nonumber\\
\hphantom{\hat{L}_{12}(u) =}{}+\bigg(\frac{k_{1}}{ \sqrt{u+j_{1}}}-{\rm i} \frac{k_{2}}{ \sqrt{u+j_{2}}}\bigg),
\\
\hat{L}_{21}(u) = \frac{\sqrt{u+j_1}\sqrt{j_2}\big(\hat{S}_{21}+ \hat{S}_{12}\big)-\sqrt{j_1}\sqrt{u+j_2}\big(\hat{S}_{21}-\hat{S}_{12}\big)}{u}\nonumber\\
\hphantom{\hat{L}_{21}(u) =}{}+ \bigg(\frac{k_{1}}{ \sqrt{u+j_{1}}}+{\rm i} \frac{k_{2}}{ \sqrt{u+j_{2}}}\bigg),
\end{gather}
\end{subequations}
where we have used that the components of the shift element $c(u)$ are given by the formula~\eqref{shelellimgl2}.

The components of spins $\hat{S}_{ij}$, $i,j = 1,2$, satisfy commutation relations~\eqref{pso}.
The corresponding integrable Hamiltonian \cite{SkrJMP2009} is obtained from the generating func\-tion~$\hat{\tau}^{(2)}(u)$:
\[
\hat{\tau}^{(2)}(u)=\frac{1}{2}\big(\hat{L}^2_{11}(u)+ \hat{L}^2_{22}(u)+ \hat{L}_{12}(u)\hat{L}_{21}(u)+ \hat{L}_{21}(u)\hat{L}_{12}(u)\big).
\]

More explicitly, we will have that one-spin Gaudin-type Hamiltonian $\hat{H}$
\[
\hat{H}= \mathrm{res}_{u=0} \hat{\tau}^{(2)}(u),
\]
 has the following explicit form:
\begin{gather}
\hat{H} =  (j_1+j_2)\big(\hat{S}_{12}\hat{S}_{21}+ \hat{S}_{21}\hat{S}_{12} \big)-
(j_1-j_2)\big( \hat{S}_{12}^2 + \hat{S}_{21}^2\big) + 2\sqrt{j_1} \sqrt{j_2} k_3 \big(\hat{S}_{11}- \hat{S}_{22}\big)\nonumber \\
\hphantom{\hat{H} =}{}+ 2 \sqrt{j_2} k_{1} \big(\hat{S}_{12}+ \hat{S}_{21}\big) + 2 i \sqrt{j_1} k_{2} \big(\hat{S}_{21}-\hat{S}_{12}\big). \label{hamGinS}
\end{gather}

\begin{Remark}Observe that the Hamiltonian~\eqref{hamGinS} is four times Hamiltonian~\eqref{hamgGc00}. This difference is explained by the renormalisation of the corresponding $r$-matrix we have performed after passing to~$\mathfrak{gl}(2)$ basis.
\end{Remark}

\begin{Remark} Observe also that in the case $k_1=k_2=0$ the Hamiltonian $\hat{H}$ coincide with the spin form of the Lipkin--Meshkov--Glick Hamiltonian and in the case partial case $k_1=k_2$ it is equivalent to the so-called ``extended'' Lipkin--Meshkov--Glick Hamiltonian \cite{RRCS}.
\end{Remark}

\section{Fermionization and the LMG model} \label{section4}
\subsection{Fermionization}
Having obtained quantum integrable spin system it is possible to
derive, using them, integrable fermion systems. For this purpose
it is necessary to consider a realization of the corresponding
spin operators in terms of fermion creation-anihilation
operators.

Let $c_{j,\s'}$, $c^{\dag}_{i,\s}$, $i, j = 1, 2,\dots , N$, $\s, \s'\in \{+,-\}$ be fermion creation-anihilation operators, i.e.,
\[
c^{\dag}_{i,\s}c_{j,\s'} +
c_{j,\s'}c^{\dag}_{i,\s}=\d_{\s\s'}\d_{ij},\qquad
c^{\dag}_{i,\s}c^{\dag}_{j,\s'} +
c^{\dag}_{j,\s'}c^{\dag}_{i,\s}=0,\qquad c_{i,\s}c_{j,\s'} +
c_{j,\s'}c_{i,\s}=0.
\]
By direct calculation it is possible to show that the following
formulae:
\begin{equation}
\hat{S}_{12}= \sum_{j=1}^N c^{\dag}_{j,+}c_{j,-}, \qquad
\hat{S}_{21}= \sum_{j=1}^N c^{\dag}_{j,-}c_{j,+}, \qquad
\hat{S}_{11}=\sum_{j=1}^N c^{\dag}_{j,+}c_{j,+}, \qquad
\hat{S}_{22}= \sum_{j=1}^N c^{\dag}_{j,-}c_{j,-} \label{fermioniz}
\end{equation}
provide realization of the Lie algebra~$\mathfrak{gl}(2)$.

\begin{Remark}
Note, that after the restriction to the
subalgebra~$\mathfrak{sl}(2)$ we obtain, substituting
$\hat{S}_{12}=\hat{S}_{+}$,
$\hat{S}_{21}=\hat{S}_{-}$,
$\hat{S}_{11}-\hat{S}_{22}=2 {\rm i} \hat{S}_{3}$, the
fermionization of the Lie algebra~$\mathfrak{sl}(2)$ given by the following formulae:
\begin{equation*}
\hat{S}_{+}= \sum_{j=1}^N c^{\dag}_{j,+}c_{j,-}, \qquad
\hat{S}_{-}= \sum_{j=1}^N c^{\dag}_{j,-}c_{j,+}, \qquad
\hat{S}_{3}=-\frac{{\rm i}}{2} \sum_{j=1}^N \big(c^{\dag}_{j,+}c_{j,+}- c^{\dag}_{j,-}c_{j,-}\big). 
\end{equation*}
\end{Remark}

\subsection{The LMG Hamiltonian}
In the partial case $k_1=k_2=0$, applying to the Hamiltonian~\eqref{hamGinS} the fermionization formulae~\eqref{fermioniz}, we obtain the following
integrable fermion Hamiltonian:
\begin{gather}
\hat{H} =  (j_1+j_2)\left( \sum_{i,j=1}^N c^{\dag}_{i,+}c_{i,-} c^{\dag}_{j,-}c_{j,+} + \sum_{i,j=1}^N c^{\dag}_{i,-}c_{i,+} c^{\dag}_{j,+}c_{j,-} \right) \nonumber\\
\hphantom{\hat{H} =}{}-(j_1-j_2)\left( \sum_{i,j=1}^N c^{\dag}_{i,+}c_{i,-} c^{\dag}_{j,+}c_{j,-}+ \sum_{i,j=1}^N c^{\dag}_{i,-}c_{i,+} c^{\dag}_{j,-}c_{j,+}\right) \nonumber\\
\hphantom{\hat{H} =}{}+ 2\sqrt{j_1} \sqrt{j_2} k_3 \sum_{j=1}^N \big(c^{\dag}_{j,+}c_{j,+}- c^{\dag}_{j,-}c_{j,-}\big). \label{hamLMG}
\end{gather}
This is the famous Lipkin--Meshkov--Glick Hamiltonian, which can -- up to a multiple of Casimir operator/ number of particle operator $\hat{N}= \sum_{i,j=1}^N \big(c^{\dag}_{i,+}c_{i,+}+ c^{\dag}_{j,-}c_{j,-}\big)$ -- be rewritten as follows:
\begin{gather}
\hat{H}_{\rm LMG} =  \frac{\epsilon}{2} \sum_{\sigma= \pm } \sum_{j=1}^N \sigma c^{\dag}_{j,\sigma}c_{j,\sigma} - \frac{W}{2} \sum_{i,j=1}^N\big( c^{\dag}_{i,+} c^{\dag}_{j,-}c_{j,+}c_{i,-} + c^{\dag}_{i,-} c^{\dag}_{j,+} c_{j,-}c_{i,+} \big)\nonumber\\
\hphantom{\hat{H}_{\rm LMG} =}{}-\frac{V}{2} \sum_{i,j=1}^N \big(c^{\dag}_{i,+}c^{\dag}_{j,+} c_{j,-}c_{i,-} + c^{\dag}_{i,-} c^{\dag}_{j,-} c_{j,+}c_{i,+}\big), \label{hamLMG'}
\end{gather}
where $\e=4 \sqrt{j_1} \sqrt{j_2} k_3$, $W=-2(j_1+j_2)$, $V=2 (j_1-j_2)$.

The Hamiltonian~\eqref{hamLMG'} is a particular example of the following general fermion Hamiltonian:
\[
\hat{H} = \frac{\epsilon}{2} \sum_{\sigma= \pm } \sum_{j=1}^N \sigma c^{\dag}_{j,\sigma}c_{j,\sigma}
+\sum_{\sigma_1, \sigma_2, \sigma_3, \sigma_4 = \pm }\, \sum_{i,j,k,l=1}^N V_{i,j,\sigma_1, \sigma_2, \sigma_3, \sigma_4} c^{\dag}_{i, \sigma_1} c^{\dag}_{j, \sigma_2}c_{j, \sigma_3} c_{i, \sigma_4}.
\]

\begin{Remark}
The case $V=0$ of the Hamiltonian~\eqref{hamLMG'}, i.e., the case $j_1=j_2$ of the Hamiltonian~\eqref{hamLMG}, is treated relatively simply due to the Cartan-invariance of the corresponding classical $r$-matrix.
 The diagonalization of the Hamiltonian~\eqref{hamLMG} in the case $j_1\neq j_2$ is more complicated. It will be performed in the next section of this article.
Our main tool will be the (modified) algebraic Bethe ansatz.
\end{Remark}

\section{Modified Bethe ansatz} \label{section5}
\subsection{The ABA basis in the Lax algebra and the gauge transformation}
In order to proceed with the algebraic Bethe ansatz it is necessary to consider a special basis in the Lax algebra~\eqref{laxalj3inf}.
In more details, let us consider the following linear functions on the Lax algebra \cite{SkrNPB2022}:
\begin{subequations}\label{ABC}
\begin{gather}
\hat{A}(u)=\frac{1}{2}\big(\hat{L}_{11}(u)+ \hat{L}_{22}(u)\big)-\frac{{\rm i} \sqrt{j_1-j_2} }{2(\sqrt{u+j_1}-\sqrt{u+j_2})} \hat{L}_{12}(u) \nonumber\\
\phantom{\hat{A}(u)=}{}+{\rm i}\frac{\sqrt{u+j_1}-\sqrt{u+j_2}}{2\sqrt{j_1-j_2}} \hat{L}_{21}(u),
\\
\hat{B}(u)=-\frac{{\rm i}}{2}\big(\hat{L}_{11}(u)-\hat{L}_{22}(u)\big)-\frac{ \sqrt{j_1-j_2} }{2(\sqrt{u+j_1}-\sqrt{u+j_2})}\hat{L}_{12}(u) \nonumber\\
\phantom{\hat{B}(u)=}{}-\frac{\sqrt{u+j_1}-\sqrt{u+j_2}}{2\sqrt{j_1-j_2}} \hat{L}_{21}(u),
\\
\hat{C}(u)=\frac{{\rm i}}{2}\big(\hat{L}_{11}(u)- \hat{L}_{22}(u)\big)-\frac{ \sqrt{j_1-j_2} }{2(\sqrt{u+j_1}-\sqrt{u+j_2})} \hat{L}_{12}(u) \nonumber\\
\phantom{\hat{C}(u)=}{}-\frac{\sqrt{u+j_1}-\sqrt{u+j_2}}{2\sqrt{j_1-j_2}} \hat{L}_{21}(u),
\\
\hat{D}(u)= \frac{1}{2}\big(\hat{L}_{11}(u)+ \hat{L}_{22}(u)\big)+\frac{{\rm i} \sqrt{j_1-j_2} }{2(\sqrt{u+j_1}+\sqrt{u+j_2})} \hat{L}_{12}(u) \nonumber\\
\phantom{\hat{D}(u)=}{}-{\rm i}\frac{\sqrt{u+j_1}+\sqrt{u+j_2}}{2\sqrt{j_1-j_2}} \hat{L}_{21}(u).
\end{gather}
\end{subequations}
In terms of these functions the generating function $\hat{\tau}^{(2)}(u)$ given by~\eqref{taugl2} is written as follows:
\begin{equation*}
\hat{\tau}^{(2)}(u)=\frac{1}{2}\bigl(\hat{A}^2(u) + \hat{D}^2(u)+ \hat{B}(u) \hat{C}(u)+ \hat{C}(u) \hat{B}(u)\bigr). 
\end{equation*}
The basis in the Lax algebra consisting of the above $A-B-C-D$ functions is the basis suitable for the algebraic Bethe ansatz. We will hereafter call it {\it ABA basis}.
It is possible to show \cite{SkrNPB2022} that
\[
\hat{A}(u)=\hat{L}_{11}^g(u), \qquad \hat{D}(u)=\hat{L}_{22}^g(u), \qquad \hat{B}(u)=\hat{L}_{21}^g(u), \qquad \hat{C}(u)=\hat{L}_{12}^g(u),
\]
where
\[
 \hat{L}^g(u)=g^{-1}(u) \hat{L}(u) g(u)
\]
and $g(u)$ is a two by two numerical matrix of the following explicit form:
\begin{equation}
g(u)=\left(
 \begin{matrix}
 -{\rm i}\dfrac{\sqrt{u+j_1}-\sqrt{u+j_2}}{\sqrt{j_1-j_2}} & \dfrac{\sqrt{u+j_1}-\sqrt{u+j_2}}{\sqrt{j_1-j_2}} \vspace{2mm}\\
 -1 & {\rm i}
 \end{matrix}
 \right), \qquad {\rm i}=\sqrt{-1}. \label{gt}
\end{equation}
 The shift element~\eqref{shelellimgl2} under the above similarity transformation acquires the following form:
\begin{gather*}
 c(u)=  -\frac{1}{\sqrt{j_1-j_2}}\bigg({\rm i}k_1 \frac{ \sqrt{u+j_{2}}}{\sqrt{u+j_{1}}}+ k_2 \frac{\sqrt{u+j_{1}}}{ \sqrt{u+j_{2}}}\bigg)(X_{11}-X_{22})\nonumber\\
\hphantom{c(u)=}{}+ \bigg( \frac{ - k_1+{\rm i} k_2}{\sqrt{j_1-j_2}}+{\rm i} k_3 \bigg)X_{12}+ \bigg( \frac{- k_1+ {\rm i} k_2}{\sqrt{j_1-j_2}}-{\rm i} k_3 \bigg)X_{21}. 
 \end{gather*}

\subsection{Lax matrix and spin Hamiltonian in the new basis}
Using the formulae~\eqref{ABC} and the explicit form of the components of one-spin Lax matrix~\eqref{onespinlax} it is easy to show that the components $\hat{L}^g_{ij}(u)$ of the transformed one-spin Lax matrix are written as follows:
\begin{subequations}\label{ABCDinT}
\begin{gather}
\hat{A}(u)=\hat{L}^g_{11}(u) = \frac{2\sqrt{u+j_1}\sqrt{u+j_2}\hat{T}_{11}}{u} -\frac{1}{\sqrt{j_1-j_2}}\bigg({\rm i}k_1 \frac{ \sqrt{u+j_{2}}}{\sqrt{u+j_{1}}}+ k_2 \frac{\sqrt{u+j_{1}}}{ \sqrt{u+j_{2}}}\bigg),
\\
\hat{D}(u)= \hat{L}^g_{22}(u) = \frac{2\sqrt{u+j_1}\sqrt{u+j_2} \hat{T}_{22}}{u}+\frac{1}{\sqrt{j_1-j_2}} \bigg({\rm i}k_1 \frac{ \sqrt{u+j_{2}}}{\sqrt{u+j_{1}}}+ k_2 \frac{\sqrt{u+j_{1}}}{ \sqrt{u+j_{2}}}\bigg),
\\
\hat{C}(u)= \hat{L}^g_{12}(u) = \frac{2\sqrt{j_1}\sqrt{j_2}\hat{T}_{21}}{u}
- {\rm i} \big(\hat{T}_{11}-\hat{T}_{22}\big) +{\rm i}\bigg( \frac{{\rm i} k_1+ k_2}{\sqrt{j_1-j_2}}+ k_3 \bigg) ,
\\
\hat{B}(u)=\hat{L}^g_{21}(u) =\frac{2\sqrt{j_1}\sqrt{j_2}\hat{T}_{12}}{u}- {\rm i} \big(\hat{T}_{11}-\hat{T}_{22}\big)+ {\rm i}\bigg( \frac{ {\rm i} k_1+ k_2}{\sqrt{j_1-j_2}}- k_3 \bigg).
\end{gather}
\end{subequations}

 The components of ``$T$-spins'' $\hat{T}_{ij}$, $i,j= 1,2$, satisfy commutation relations of~$\mathfrak{gl}(2)$
\[
\big[\hat{T}_{ij}, \hat{T}_{kl}\big]= \d_{kj} \hat{T}_{il}- \d_{il} \hat{T}_{kj}.
\]
They coincide with the components of the gauge-transformed ``$S$-spin'' $\hat{S}$, where the transformation matrix is given by~\eqref{gt} with $u = 0$. In more details, we have
\begin{subequations}\label{trnsfS-T}
\begin{gather}
\hat{T}_{11}= \frac{1}{2}\big(\hat{S}_{11}+\hat{S}_{22}\big)+\frac{{\rm i}}{2}\frac{\sqrt{j_1}- \sqrt{j_2}}{\sqrt{j_1-j_2}}\hat{S}_{12}- \frac{{\rm i}}{2} \frac{\sqrt{j_1-j_2}}{\sqrt{j_1}- \sqrt{j_2}} \hat{S}_{21},
\\
\hat{T}_{22}= \frac{1}{2}\big(\hat{S}_{11}+\hat{S}_{22}\big)-\frac{{\rm i}}{2}\frac{\sqrt{j_1}- \sqrt{j_2}}{\sqrt{j_1-j_2}}\hat{S}_{12} + \frac{{\rm i}}{2} \frac{\sqrt{j_1-j_2}}{\sqrt{j_1}- \sqrt{j_2}} \hat{S}_{21},
\\
\hat{T}_{21}= \frac{{\rm i}}{2}\big(\hat{S}_{11}- \hat{S}_{22}\big)-\frac{1}{2}\frac{\sqrt{j_1}- \sqrt{j_2}}{\sqrt{j_1-j_2}}\hat{S}_{12}- \frac{1}{2} \frac{\sqrt{j_1-j_2}}{\sqrt{j_1}- \sqrt{j_2}} \hat{S}_{21},
\\
\hat{T}_{12}= -\frac{{\rm i}}{2}\big(\hat{S}_{11}- \hat{S}_{22}\big)-\frac{1}{2}\frac{\sqrt{j_1}- \sqrt{j_2}}{\sqrt{j_1-j_2}}\hat{S}_{12}- \frac{1}{2} \frac{\sqrt{j_1-j_2}}{\sqrt{j_1}- \sqrt{j_2}} \hat{S}_{21}.
\end{gather}
\end{subequations}

The corresponding spin Hamiltonian in the external magnetic field
\[
\hat{H}=\mathrm{res}_{u=0} \hat{\tau}^{(2)}(u)
\]
is -- up to Casimir operator $\hat{C}^{(2)}=\hat{T}_{11}^2 + \hat{T}_{22}^2+\hat{T}_{12} \hat{T}_{21}+ \hat{T}_{21} \hat{T}_{12}$ -- explicitly written as follows:
\begin{gather}
\hat{H}=
 -{\rm i} \sqrt{j_1}\sqrt{j_2} \big( \big(\hat{T}_{11}-\hat{T}_{22}\big)\big( \hat{T}_{12}+ \hat{T}_{21}\big)+ \big( \hat{T}_{12}+ \hat{T}_{21}\big)\big(\hat{T}_{11}-\hat{T}_{22}\big)\big)\nonumber\\
\hphantom{\hat{H}=}{}+ 2(j_1+j_2)\big(\hat{T}_{11}^2+ \hat{T}_{22}^2\big)- \frac{2}{\sqrt{j_1-j_2}}({\rm i}k_1 j_2+ k_2 j_1)\big(\hat{T}_{11}-\hat{T}_{22}\big)\nonumber\\
\hphantom{\hat{H}=}{}+ 2 {\rm i} \sqrt{j_1}\sqrt{j_2}\bigg( \frac{{\rm i} k_1+ k_2}{\sqrt{j_1-j_2}}+ k_3 \bigg)\hat{T}_{12}+ 2 {\rm i} \sqrt{j_1}\sqrt{j_2}\bigg( \frac{ {\rm i} k_1+ k_2}{\sqrt{j_1-j_2}}- k_3 \bigg) \hat{T}_{21}. \label{hamgt}
\end{gather}

Applying the transform \eqref{trnsfS-T} to the Hamiltonian~\eqref{hamgt} we recover -- modulo the square of the linear
Casimir operators $\hat{I}= \hat{S}_{11}+\hat{S}_{22}$
 -- the spin Hamiltonian~\eqref{hamGinS}.

Observe that in the case $k_1=k_2=0$ the Hamiltonian~\eqref{hamgt} acquires the following simple form:
\begin{gather}
\hat{H}= 2 {\rm i} \sqrt{j_1}\sqrt{j_2}(k_3+1) \big(\hat{T}_{12}- \hat{T}_{21}\big) -2 {\rm i} \sqrt{j_1}\sqrt{j_2} \big(\hat{T}_{11}-\hat{T}_{22}\big)\big( \hat{T}_{12}+ \hat{T}_{21}\big)\nonumber\\
\hphantom{\hat{H}=}{}
+2(j_1+j_2)\big(\hat{T}_{11}^2+ \hat{T}_{22}^2\big). \label{hamgt'}
\end{gather}

\subsubsection{The representation space}
In order to apply the results of the previous subsections to the considered spin model in an external field we need to describe the space of representation of the corresponding Lax algebra~\eqref{ABCDinT}.
For this purpose we consider an irreducible finite-dimensional irreducible representation of the algebra~$\mathfrak{gl}(2)$ of ``$T$-spins'' in some space $V$. Each such the representation $V$ contains the lowest and highest weight vectors $\mathrm{v}_+$ and $\mathrm{v}_-$ such that
\begin{alignat*}{4}
& \hat{T}_{11}
\mathrm{v}_+=\l_1 \mathrm{v}_+, \qquad&& \hat{T}_{22}\mathrm{v}_+ = \l_2 \mathrm{v}_+, \qquad&& \hat{T}_{21} \mathrm{v}_{+}=0,& 
\\
& \hat{T}_{11}
\mathrm{v}^-=\l_2 \mathrm{v}^-, \qquad&& \hat{T}_{22}\mathrm{v}^- = \l_1 \mathrm{v}^-, \qquad&& \hat{T}_{12} \mathrm{v}^{-}=0.& 
\end{alignat*}
The space $V$ is spanned by the vectors
\[
\mathrm{v}_{+}^{(m)}=(\hat{T}_{12})^m \mathrm{v}_+, \qquad
m=0,\dots , (\l_2 - \l_1),
\]
 or by the vectors
 \[
 \mathrm{v}_{-}^{(m)}=(\hat{T}_{21})^m \mathrm{v}_-, \qquad
m=0,\dots , (\l_2 - \l_1).
\]
The dimension of the space $V$ is hence equal to $n_{\l}=(\l_2 - \l_1)+1$.

\begin{Remark} We will hereafter put $\l_2=N$, $\l_1=0$. This will correspond to the considered below fermion representation, where the vector $\mathrm{v}_+$ is annulled by the fermion operators $f_{j,+}$, $f^{\dag}_{j,-}$, $j = 1,\dots ,N$, and the vector $\mathrm{v}_-$ is annulled by the fermion operators $f_{j,-}$, $f^{\dag}_{j,+}$, $j = 1,\dots , N$.
\end{Remark}

\subsection{The fermionic LMG model in the new basis}
\subsubsection{Second fermionization and the canonical transformation}
Similar to the previous subsection, we will need to construct the
fermionization of the Lie algebra~$\mathfrak{gl}(2)$ in the new basis, i.e., to fermionize the operators $\hat{T}_{ij}$.
The fermionization will be written with the help of another group of fermion creation-anihilation operators $f_{j,\s'}$, $f^{\dag}_{i,\s}$, $i,j = 1,\dots ,N$, $\s,\s'\in
\{+,-\}$ satisfying the same Clifford algebra of anti-commutation relations:
\[
f^{\dag}_{i,\s}f_{j,\s'} +
f_{j,\s'}f^{\dag}_{i,\s}=\d_{\s\s'}\d_{ij},\qquad
f^{\dag}_{i,\s}f^{\dag}_{j,\s'} +
f^{\dag}_{j,\s'}f^{\dag}_{i,\s}=0,\qquad f_{i,\s}f_{j,\s'} +
f_{j,\s'}f_{i,\s}=0.
\]
The same type of
formulae provide the realization of the Lie algebra~$\mathfrak{gl}(2)$:
\begin{alignat}{3}
&\hat{T}_{12}= \sum_{j=1}^N f^{\dag}_{j,+}f_{j,-}, \qquad&&
\hat{T}_{21}= \sum_{j=1}^N f^{\dag}_{j,-}f_{j,+},& \nonumber\\
&\hat{T}_{11}=\sum_{j=1}^N f^{\dag}_{j,+}f_{j,+}, \qquad&&
\hat{T}_{22}= \sum_{j=1}^N f^{\dag}_{j,-}f_{j,-}.& \label{fermioniz''}
\end{alignat}

The following proposition holds true:
\begin{Proposition}
Let the operators $\hat{T}_{ij}$ and $\hat{S}_{kl}$, $m = 1,\dots , N$, be connected by the transformation~\eqref{trnsfS-T}.
Then $f$- and $c$-types of fermions are connected by the following canonical transformation:
\begin{subequations}\label{ferm-transf}
\begin{alignat}{3}
&f_{m,+}=- \frac{{\rm i}}{2} \frac{\sqrt{j_1-j_2}}{\sqrt{j_1}-\sqrt{j_2}} c_{m,+}+ \frac{1}{2} c_{m,-}, \qquad&&
f^{\dag}_{m,+}= {\rm i}\frac{\sqrt{j_1}-\sqrt{j_2}}{\sqrt{j_1-j_2}} c^{\dag}_{m,+} + c^{\dag}_{m,-},&
\\
&f_{m,-}=- \frac{1}{2} \frac{\sqrt{j_1-j_2}}{\sqrt{j_1}-\sqrt{j_2}} c_{m,+}+ \frac{{\rm i}}{2} c_{m,-}, \qquad&&
 f^{\dag}_{m,-}= - \frac{\sqrt{j_1}-\sqrt{j_2}}{\sqrt{j_1-j_2}} c^{\dag}_{m,+} - {\rm i} c^{\dag}_{m,-}.&
\end{alignat}
\end{subequations}
\end{Proposition}

\begin{proof}The proof is achieved by straightforward verification.
\end{proof}
\subsubsection{ LMG Hamiltonian in the new basis}
Using the fermionization formulae~\eqref{fermioniz''} and the explicit form of the Hamiltonian~\eqref{hamgt'} we obtain the following simple form of the LMG Hamiltonian~\eqref{hamgt'} in the new fermion basis:
\begin{gather}
\hat{H}=  2 {\rm i} \sqrt{j_1}\sqrt{j_2}(k_3+1) \sum_{j=1}^N \big(f^{\dag}_{j,+}f_{j,-}- f^{\dag}_{j,-}f_{j,+}\big)-2 {\rm i} \sqrt{j_1}\sqrt{j_2} \sum_{j=1}^N \big(f^{\dag}_{j,+}f_{j,+}- f^{\dag}_{j,-}f_{j,-}\big) \nonumber\\
\hphantom{\hat{H}=}{}\times \sum_{j=1}^N \big(f^{\dag}_{j,+}f_{j,-} + f^{\dag}_{j,-}f_{j,+}\big)+2(j_1+j_2) \Bigg(\Bigg(\sum_{j=1}^N f^{\dag}_{j,+}f_{j,+}\Bigg)^2+ \Bigg(\sum_{j=1}^N f^{\dag}_{j,-}f_{j,-}\Bigg)^2\Bigg),\!\! \label{hamgtf}
\end{gather}

\begin{Remark}
 Observe that the Hamiltonian~\eqref{hamgtf} coincide -- up to the transform~\eqref{ferm-transf} and identity operators -- with the Hamiltonian~\eqref{hamLMG}.
That is why, diagonalizing the Hamiltonian~\eqref{hamgtf} we will automatically diagonalize also the Hamiltonian~\eqref{hamLMG} and vice versa. \end{Remark}

\subsection{The spectrum and Bethe equations: the general case}
The main ingredient of the algebraic Bethe ansatz are the Bethe vectors.
For the definition of the Bethe vectors we will use the following products \cite{SkrNPB2022}:
\begin{equation*}
\hat{\mathcal{B}}(v_1, v_2, \dots , v_N)=\hat{B}_{1}(v_1) \hat{B}_{2}(v_2) \cdots \hat{B}_{N}(v_N), 
\end{equation*}
where
\begin{equation*}
\hat{B}_{k}(v)= \hat{B}(v)-{\rm i}(2k-1) \operatorname{Id}, \qquad k\in \mathbb{Z}, \qquad {\rm i}=\sqrt{-1}, 
\end{equation*}
and $v_i$, $i = 1,\dots , N$, are ``rapidities'' -- complex numbers to be determined from the Bethe-type equations.
Due to the symmetry of the structure of the Lax algebra we will also consider the products
\begin{equation*}
\hat{\mathcal{C}}(v_1, v_2, \dots , v_N)=\hat{C}_{1}(v_1) \hat{C}_{2}(v_2) \cdots \hat{C}_{N}(v_N), 
\end{equation*}
where
\begin{equation*}
\hat{C}_{k}(v)= \hat{C}(v) + {\rm i} (2k-1) \operatorname{Id}, \qquad k\in \mathbb{Z}, \qquad {\rm i}=\sqrt{-1}, 
\end{equation*}
and the values of rapidities $v_i$, $i = 1,\dots , N$, have to be determined from the other Bethe-type equations.

The following theorem is a consequence of the more general theorem of \cite{SkrNPB2022}:
\begin{Theorem}\quad
\begin{enumerate}\itemsep=0pt
\item[$(i)$] Let
$k_{\pm}=\frac{{\rm i} k_1+ k_2}{\sqrt{j_1-j_2}} \pm k_3$ and $k_-\not = N+1, N+3, \dots , 3N+1$. Let the rapidities $v_k$, $k = 1,\dots ,N$, satisfy Bethe-type equations
\begin{align}
 &{}-\frac{N (v_k+j_1)(v_k+j_2)}{v_k} -\frac{{\rm i}k_1 (v_k+j_{2})+ k_2 (v_k+j_{1})}{\sqrt{j_1-j_2}}\nonumber\\
 &\qquad{}+(2 v_k+j_1+j_2)+ 2\sum_{n=1,\, n\neq k}^N \frac{(v_k+j_1)(v_k+j_2)}{v_k-v_n}\nonumber\\
 &\qquad\quad{}= \frac{1}{4}(k_+ - N -1)(k_- - N -1) \frac{v_k^{N}}{\prod\limits_{n=1,\, n\neq k}^{N}(v_k-v_n)}, \qquad k = {1,\dots ,N}. \label{betheeq1'}
\end{align}
Then the following vectors
in the space $V$
\[
V_+(v_1, v_2, \dots , v_N)= \hat{\mathcal{B}}(v_1, v_2, \dots , v_N) \mathrm{v}_{+}
\]
 are the eigen-vectors of quantum Hamiltonian $\hat{H}$ with eigenvalues
\begin{equation*}
h_+(v_1, v_2, \dots , v_N)=2 N \left((j_1+j_2) N +\frac{{\rm i}k_1 j_{2}+ k_2 j_{1}}{\sqrt{j_1-j_2}}+2 j_1 j_2 \sum_{n=1}^N \frac{1}{v_n} \right). 
\end{equation*}

\item[$(ii)$] Let
 $k_{\pm}=\frac{{\rm i} k_1+ k_2}{\sqrt{j_1-j_2}}\pm k_3$ and $k_+\not= -(N+1), -(N+3), \dots , -(3N+1)$. Let the rapidities~$v_k$, $k = 1,\dots , N$, satisfy Bethe-type equations
\begin{align}
 &{}-\frac{N (v_k+j_1)(v_k+j_2)}{v_k} + \frac{{\rm i}k_1 (v_k+j_{2})+ k_2 (v_k+j_{1})}{\sqrt{j_1-j_2}} \nonumber\\
 &\qquad{}+ (2 v_k+j_1+j_2) +2\sum_{n=1,\, n\neq k}^N \frac{(v_k+j_1)(v_k+j_2)}{v_k-v_n}\nonumber\\
 &\qquad\quad{}= \frac{1}{4}(k_+ + N+1)(k_- + N+1) \frac{v_k^{N}}{\prod\limits_{n=1,\, n\neq k}^{N}(v_k-v_n)}, \qquad k=1,\dots , N. \label{betheeq2'}
\end{align}
Then the following vectors
in the space $V$
\[
V_-(v_1, v_2, \dots , v_N)= \hat{\mathcal{C}}(v_1, v_2, \dots , v_N) \mathrm{v}_{-}
\]
 are the eigen-vectors of quantum Hamiltonian $\hat{H}$ with eigenvalues
 \begin{equation}
h_-(v_1, v_2, \dots , v_N)=2 N \left(\big(j_1+j_2\big) N - \frac{{\rm i}k_1 j_{2}+ k_2 j_{1}}{\sqrt{j_1-j_2}}+2 j_1 j_2 \sum_{n=1}^N \frac{1}{v_n} \right). \label{spehm}
\end{equation}
\end{enumerate}
\end{Theorem}

\begin{Remark}
 Observe that the exclusion of certain values of $k_{\pm}$ in the theorem is due to our way of proving it \cite{SkrNPB2022}. The statement of the theorem seem to be true for any values of $k_{\pm}$, but we have the proof in the cases $k_-\not= N+1, N+3, \dots , 3N+1$ and $k_+\not=  -(N+1), -(N+3),\allowbreak \dots , -(3N+1)$ respectively.
\end{Remark}

\begin{Remark}
 Note that consideration of small $N$ examples shows that for the generic values of~$k_i$, $i = 1,2,3$, both systems of the eigenvectors $\{V_+(v_1, v_2, \dots , v_N)\}$ and $\{V_-(v_1, v_2, \dots , v_N)\}$ are the same. The sets of the corresponding eigenvalues $\{h_+(v_1, v_2, \dots , v_N)\}$ and $\{h_-(v_1, v_2, \dots ,\allowbreak  v_N)\}$ also coincide.
\end{Remark}

\subsection[Example: N=1 case]{Example: $\boldsymbol{N=1}$ case}
Let us consider the simplest case $N=1$ corresponding to two-dimensional representation of~$\mathfrak{sl}(2)$ and~$\mathfrak{gl}(2)$. In this case there is only one rapidity $v_1$ and the Bethe equation~\eqref{betheeq1'} acquires the form
\begin{gather}
 -\frac{ (v_1+j_1)(v_1+j_2)}{v_1} -\frac{{\rm i}k_1 (v_1+j_{2})+ k_2 (v_1+j_{1})}{\sqrt{j_1-j_2}}\nonumber\\
 \qquad{}+(2 v_1+j_1+j_2)- \frac{1}{4}(k_+ - 2)(k_- -2){v_1}=0, \label{betheeq1-1}
\end{gather}
where, as previously, $k_{\pm}=\frac{{\rm i} k_1+ k_2}{\sqrt{j_1-j_2}} \pm k_3$.

The eigenvalues of the Hamiltonian $\hat{H}$ given by~\eqref{hamgt} are written as follows:
\begin{equation}
h_+(v_1)=2 \left((j_1+j_2) +\frac{{\rm i}k_1 j_{2}+ k_2 j_{1}}{\sqrt{j_1-j_2}}+ \frac{2 j_1 j_2}{v_1} \right). \label{spehp-1}
\end{equation}

Introducing the variable $x=v_1^{-1}$ we rewrite the equation~\eqref{betheeq1-1} as follows:
\begin{equation}
 -j_1j_2x^2-\frac{({\rm i}k_1j_2+k_2j_1)x}{\sqrt{j_1-j_2}}+\frac{1}{4}k_3^2+ \frac{1}{4}\frac{(-k_1+ik_2)^2}{(j_1-j_2)}=0. \label{betheeq1-1'}
 \end{equation}
The equation~\eqref{betheeq1-1'} has two solutions
\begin{gather*}
x_{1} = -\frac{1}{2 j_1j_2}\left(\frac{{\rm i}k_1j_2+k_2j_1}{\sqrt{j_1-j_2}}-\sqrt{j_2k_1^2+j_1k_2^2+j_1j_2k_3^2}\right),
\\
x_2 = -\frac{1}{2 j_1j_2}\left(\frac{{\rm i}k_1j_2+k_2j_1}{\sqrt{j_1-j_2}}+\sqrt{j_2k_1^2+j_1k_2^2+j_1j_2k_3^2}\right).
\end{gather*}
Substituting them into the formula~\eqref{spehp-1} we obtain the following eigenvalues of $\hat{H}$:
\[
h_1 = 2 j_1+2 j_2+2 \sqrt{j_2 k_1^2+j_1 k_2^2+j_1 j_2 k_3^2}, \qquad h_2 = 2 j_1+2 j_2-2 \sqrt{j_2 k_1^2+j_1 k_2^2+j_1 j_2 k_3^2}.
\]

We remark that the formula~\eqref{spehm} and Bethe equations~\eqref{betheeq2'} produce in the generic case the same spectrum, i.e., the vectors $V_+(v_1)$ and $V_-(v_1)$ for the generic values of $k_i$, $i= 1,2,3$, are proportional.

We also remark, that in the considered two by two matrix representation the eigenvalues of~\eqref{hamgt} are easily calculated by the direct method
and are shown to coincide with the $h_1$, $h_2$ written above.

\subsection[Example: N=2 case]{Example: $\boldsymbol{N=2}$ case}
Let us consider the case $N=2$ corresponding to three-dimensional representation of~$\mathfrak{sl}(2)$ and~$\mathfrak{gl}(2)$. In this case there are two rapidities $v_1$ and $v_2$ and the Bethe equations~\eqref{betheeq1'} acquire the form
\begin{subequations} \label{betheeq1-2}
\begin{gather}
 -\frac{2 (v_1+j_1)(v_1+j_2)}{v_1} -\frac{{\rm i}k_1 (v_1+j_{2})+ k_2 (v_1+j_{1})}{\sqrt{j_1-j_2}}+(2 v_1+j_1+j_2) + 2 \frac{(v_1+j_1)(v_1+j_2)}{v_1-v_2}\nonumber\\
\qquad{}= \frac{1}{4}(k_+ - 3)(k_- - 3) \frac{v_1^{2}}{(v_1-v_2)}, \label{betheeq1-2-1}
\\
 -\frac{2 (v_2+j_1)(v_2+j_2)}{v_2} -\frac{{\rm i}k_1 (v_2+j_{2})+ k_2 (v_2+j_{1})}{\sqrt{j_1-j_2}}+(2 v_2+j_1+j_2) + 2 \frac{(v_2+j_1)(v_2+j_2)}{v_2-v_1}\nonumber\\
\qquad{}= \frac{1}{4}(k_+ - 3)(k_- - 3) \frac{v_2^{2}}{(v_2-v_1)}, \label{betheeq1-2-2}
\end{gather}
\end{subequations}
where, as previously, $k_{\pm}=\frac{{\rm i} k_1+ k_2}{\sqrt{j_1-j_2}} \pm k_3$.

The eigenvalues of the Hamiltonian $\hat{H}$ given by~\eqref{hamgt} are written as follows:
\begin{equation*}
h_+=4 \left(2(j_1+j_2) +\frac{{\rm i}k_1 j_{2}+ k_2 j_{1}}{\sqrt{j_1-j_2}}+ 2 j_1 j_2\bigg(\frac{1}{v_1} + \frac{1}{v_2}\bigg) \right). 
\end{equation*}
In order to solve the equations~\eqref{betheeq1-2} and to find the spectrum we at first make a change of variables
 $X=\big(v^{-1}_1+v_2^{-1}\big)$, $Y= \big(v^{-1}_1-v_2^{-1}\big)$. Then we consider the sum and the difference of the equations~\eqref{betheeq1-2-1},~\eqref{betheeq1-2-2} and obtain the following equations for $X$, $Y$:
\begin{subequations} 
\begin{gather}
2 j_1 j_2 (j_1-j_2) X^3+
4 (j_1-j_2) (j_1+j_2) X^2+\bigl(6 j_1 j_2 (j_1-j_2) Y^2+(j_1-j_2) k_3^2+ (j_2- j_1)\nonumber\\
\qquad{}+(k_1- {\rm i}k_2)^2+6 ({\rm i} k_1+k_2) \sqrt{j_1-j_2}\bigr) X+4 ({\rm i} k_1 j_2+k_2 j_1) \sqrt{j_1-j_2} Y^2=0,\label{betheeq1-2-1'} \\
Y^2+3 X^2+ \frac{(2 \sqrt{j_1-j_2} ({\rm i}k_1 j_2+k_2 j_1)+j_1^2-j_2^2) }{j_1 j_2 (j_1-j_2)} X-
\frac{k_3^2}{2 j_1 j_2}\nonumber\\
\qquad{}+ \frac{(2 ({\rm i} k_1+k_2) \sqrt{j_1-j_2}+j_1-j_2 -(k_1-{\rm i}k_2)^2)}{2 j_1 j_2 (j_1-j_2)}=0. \label{betheeq1-2-2'}
\end{gather}
\end{subequations}
Resolving the equation~\eqref{betheeq1-2-2'} with respect to $Y^2$, substituting it into~\eqref{betheeq1-2-1'}, taking into account that
\[
 X =\frac{h_+}{8 j_1j_2}- \bigg(\frac{(j_1+j_2)}{ j_1 j_2} +\frac{{\rm i}k_1 j_{2}+ k_2 j_{1}}{2 j_1j_2\sqrt{j_1-j_2}}\bigg),
\]
we finally derive the following equation for the eigenvalues $h_+$:
\begin{gather}
h_+^3-20(j_2+ j_1) h_+^2+\big(128 j_1^2-16 j_1 k_3^2 j_2+272 j_1 j_2+128 j_2^2-16 j_1 k_2^2-16 j_2 k_1^2\big) h_+ \nonumber\\
\qquad{} +128 (j_1+j_2) j_1 j_2 k_3^2
+64 j_2 (2 j_1+j_2) k_1^2+64 j_1 (j_1+2 j_2) k_2^2\nonumber\\
\qquad{}-128 (j_1+2 j_2) (2 j_1+j_2) (j_1+j_2)=0. \label{charpol2}
\end{gather}

We remark, that in the considered three-dimensional representation the characteristic polynomial of the Hamiltonian~\eqref{hamgt} is easily calculated by direct method and coincide with the polynomial~\eqref{charpol2}.

We also remark that the formula~\eqref{spehm} and Bethe equations~\eqref{betheeq2'} produce in the generic case the same spectrum, i.e., the vectors $V_+(v_1,v_2)$ and $V_-(v_1,v_2)$ for generic values of $k_i$, $i= 1,2,3$, are proportional.

\section{Conclusion and discussion} \label{section6}
In the present paper we have shown that the Lipkin--Meshkov--Glick $2 N$-fermion model is a~particular case of one-spin Gaudin-type model in an external magnetic field based on the non-skew-symmetric classical elliptic $r$-matrix in the $j_3\rightarrow \infty$ limit. We have also constructed a~further integrable generalization of the Lipkin--Meshkov--Glick model, which, written in the spin form, coincides with the quantum Zhukovsky--Volterra gyrostat. We have diagonalized the corresponding quantum Hamiltonian by means of the modified algebraic Bethe ansatz and explicitly solved the corresponding Bethe-type equations for small number of fermions $N=1,2$.

It will be very useful to find -- at least numerically -- the solutions of the constructed Bethe-type equations for the case of the arbitrary $N$. This problem is open. It will be also interesting to compare the obtained Bethe ansatz with other variants of the Bethe ansatz, existing in the literature for the (non-generalized) Lipkin--Meshkov--Glick model \cite{Duk, Pan}.

\subsection*{Acknowledgements}
The author is grateful to N. Manojlovic for the Bethe ansatz discussions.

\pdfbookmark[1]{References}{ref}
\LastPageEnding

\end{document}